\documentclass{llncs}

\usepackage{amssymb}           %% \varsubsetneq
\usepackage{stmaryrd}          %% \bigsqcap
\usepackage[dvips]{graphicx}   %% needed for rotation in order to get the Par symbol
\usepackage[matrix,arrow]{xy}  %% just one diagram

\title{Predicate transformers and Linear Logic}
\subtitle{yet another Denotational Model}

\author{Pierre Hyvernat\inst{1,2}}
\institute{Institut math\'ematique de Luminy, Marseille, France\\
           \and
           Chalmers Institute of Technology, G\"oteborg, Sweden\\
           \email{hyvernat@iml.univ-mrs.fr}}

\makeatletter
\let\ea=\expandafter

\def\tvi{\vrule height12pt depth5pt width0pt}

\newbox\hyp \newbox\concl \newdimen\ruleWidth

\def\rule@nd{\rule@end}

\def\rule@ux#1&{
 \def\tmp{#1}
  \ifx\tmp\rule@nd
   \hskip-2em$
    \else\tmp\hskip2em\ea\rule@ux
     \fi}

\def\infer#1#2#3{\relax
 \setbox\hyp=\hbox{$\rule@ux#1&\rule@end&} \setbox\concl=\hbox{$#2$}
  \ifdim \wd\concl<\wd\hyp
   \ruleWidth=\wd\hyp
    \else\ruleWidth=\wd\concl
     \fi
     \advance\ruleWidth by 1cm
     $\vcenter{
       \vbox{\hbox to\ruleWidth{\hss\tvi \unhbox\hyp \hss}}
        \hrule height.7pt depth0pt width\ruleWidth
         \vbox{\hbox to\ruleWidth{\hss\tvi \unhbox\concl \hss}}}$
          \kern.1cm \hbox{#3}}
\makeatother
\newcommand\ie{\textit{i.e.}~}
\newcommand\st{\hbox{\rm s.t.}~}
\newcommand\ala{\hbox{\sl \`a la}~}

\newcommand\sitem[1]{\item[\small{\textit{(#1)}}]}
\newcommand\sit[1]{{\small\textit{(#1)}}}

\newcommand{\Prop}{\mathrm{P\kern-1pt r}}
\newcommand{\Set}{\hbox{\bf Set}}        % categories: Set
\newcommand{\Rel}{\hbox{\bf Rel}}        %             Rel
\newcommand{\POW}{\hbox{\bf Pow}}        %             Pow
\newcommand{\Int}{\hbox{\bf Int}}        %             Spec (interfaces)

\newcommand{\Pow}{\mathcal{P}}           % powerset
\newcommand{\Mulf}{\mathcal{M}_{\!f}}    % finite multisets
\newcommand{\Sd}{\mathcal{S}}            % seeds

\newcommand\compl[1]{\overline{#1}}      % complement with
\newcommand\Id{\hbox{\bf Id}}       % identity

\newcommand\pr{\ast}
\renewcommand\Pr{\bigsqcap}

\newcommand\step[2]{\hskip.25cm$#1$\ifx\empty#2\else\hskip.25cm$\{$ {\small #2} $\}$\fi}

\mathcode`\?="003F
\mathcode`\!="0021
\mathcode`\&="2026
\newcommand\tensor\otimes
\newcommand\plus\oplus
\newcommand\with{\mathbin\&}
\newcommand\parsym{\rotatebox[origin=c]{180}{\ensuremath{\&}}}
\newcommand\Par{\mathbin{\parsym}}
\newcommand\linearArrow{\mathbin{\relbar\mskip-8mu\circ}}
\newcommand\bone{\ensuremath\mathbf{1}}
\newcommand\bzero{\ensuremath\mathbf{0}}

\newcommand\p{\raise.7pt\hbox{$\scriptstyle+$}}
\newcommand\m{\raise.7pt\hbox{$\scriptstyle-$}}
\newcommand\magic{\mathsf{magic}}
\newcommand\PI{\mathsf{\Pi}}
\begin{document}

%%%%%%%%%%%%%%%%%%%%%%%%%%%%%%%%%%%%%%%%%%%%%%%%%%%%%%%%%%%%%%%%%%%%%%%%%%%%%%
%%% title / abstract <<<1
\maketitle

\begin{abstract}
  In the refinement calculus, monotonic predicate transformers are used to
  model specifications for (imperative) programs. Together with a natural
  notion of simulation, they form a category enjoying many algebraic
  properties.

  We build on this structure to make predicate transformers into a
  denotational model of full linear logic: all the logical constructions have
  a natural interpretation in terms of predicate transformers (\ie in terms of
  specifications). We then interpret proofs of a formula by a safety property
  for the corresponding specification.
\end{abstract}
%%% >>>1

%%%%%%%%%%%%%%%%%%%%%%%%%%%%%%%%%%%%%%%%%%%%%%%%%%%%%%%%%%%%%%%%%%%%%%%%%%%%%%
\section*{Introduction} %%%<<<1

The first denotational model for linear logic was the category of
\emph{coherent spaces} (\cite{LL}). In this model, formulas are interpreted by
graphs; and proofs by \emph{cliques} (complete subgraphs). This forms a special
case of domain \ala Scott.

From a conceptual point of view, the construction of interfaces is a little
different: first, the model looks a little more dynamic; then, \emph{seeds}
---the notion corresponding to cliques--- are not closed under substructures;
and finally, they are closed under arbitrary unions (usually, only directed
unions are allowed).

What was a little unexpected is that the interpretation of linear proofs used
in the relational model can be lifted directly to this structure to yield a
denotational model of full linear logic in the spirit of
$\_$/hyper/multi-coherence or finiteness spaces.

A promising direction for further research is to explore the links between the
model presented below and non-determinism as it appears both in the
differential lambda-calculus (\cite{difflambda,diffnet}) and different kind of
process calculi. We expect such a link because of the following remarks: this
model comes from the semantics of imperative languages; it can be extended to
a model of the differential lambda calculus (which can be seen as a variant of
``lambda calculus with resource'') and there is a completely isomorphic
category in which predicate transformers are replaced by (two-sided)
transition systems. In particular, all of the logical operations presented
below have natural interpretations in terms of processes...
%%%>>>1

%%%%%%%%%%%%%%%%%%%%%%%%%%%%%%%%%%%%%%%%%%%%%%%%%%%%%%%%%%%%%%%%%%%%%%%%%%%%%%
\section{Relations and Predicate Transformers} %%%<<<1

\begin{definition}
  A relation $r$ between two sets is a subset of their cartesian product. We
  write $r^\sim$ for the \emph{converse} relation: $r^\sim = \big\{ (b,a)\ |\
  (a,b)\in r\big\}$.
\\
  The composition of two relations $r\subseteq A\times B$ and $r'\subseteq
  B\times C$ is defined by $r'\cdot r = \big\{ (a,c)\ |\ (\exists b\in B)\
  (a,b)\in r \land (b,c)\in r'\big\}$.
\\
  If $X$ is a set, $\Id_X$ denotes the identity on $X$, \ie $\Id_X=\{(a,a)\ |\
  a\in X\}$.
\end{definition}

There seems to be three main notions of morphisms between sets. These give rise
to three important categories in computer science:
\begin{itemize}
  \item $\Set$, where morphisms are functions;
  \item $\Rel$, where morphisms are (binary) relations;
  \item $\POW$, where morphisms are monotonic \emph{predicate transformers.}
\end{itemize}
One can go from $\Set$ to $\Rel$ and from $\Rel$ to $\POW$ using the same
categorical construction (\cite{newalg}) which cannot be applied further.

\begin{definition}
  A {predicate transformer} from $A$ to $B$ is a function from $\Pow(A)$ to
  $\Pow(B)$.  A predicate transformer $P$ is {monotonic} if $x\subseteq x'$
  implies $P(x)\subseteq P(x')$.
\end{definition}
From now on, we will consider only monotonic predicate transformers. The
adjective ``monotonic'' is thus implicit everywhere.

The term ``predicate'' might not be the most adequate but the terminology was
introduced by E. Dijkstra some decades ago, and has been used extensively by
computer scientists since then. Formally, a predicate on a set $A$ can be
identified with a subset of $A$ by the separation axiom of ZF set theory; the
confusion is thus harmless.

\begin{definition}
  If $r$ is a relation between $A$ and $B$, we write $\langle r\rangle
  :\Pow(A)\to\Pow(B)$ for the following predicate transformer:
  {\hskip.5cm\small(called the direct image of $r$)}
  \[
    \langle r\rangle (x) = \big\{ b\in B\ |\ (\exists a\in A)\ (a,b)\in r\land
    a\in x\big\} \ \hbox{.}
  \]
\end{definition}
Note that in the traditional version of the refinement calculus (\cite{RC}),
our $\langle r\rangle$ is written $\{r^\sim\}$, but this notation clashes with
set theoretic notation and would make our formulas very verbose with $\_^\sim$
everywhere.
%%%>>>1

%%%%%%%%%%%%%%%%%%%%%%%%%%%%%%%%%%%%%%%%%%%%%%%%%%%%%%%%%%%%%%%%%%%%%%%%%%%%%%
\section{Interfaces} %%%<<<1

Several denotational models of linear logic can be seen as ``refinements'' of
the relational model. This very crude model interprets formulas by sets; and
proofs by subsets. It is degenerate in the sense that any formula is
identified with its linear negation! Coherent spaces (\cite{LL}),
hypercoherent spaces (\cite{hyper}), finiteness spaces (\cite{finiteness})
remove (part of) this degeneracy by adding structure on top of the relational
model.  We follow the same approach:
\begin{definition} An \emph{interface} $X$ is given by a set $|X|$ (called the
  \emph{state space}) and a predicate transformer $P_X$ on $|X|$ (called the
  \emph{specification}).
\end{definition}
The term ``specification'' comes from computer science, where a specification
usually takes the form:

\bgroup\parindent=2cm\narrower\sl\noindent
  if the program is started in a state satisfying $\phi$, it will terminate;
  and the final state will satisfy $\psi$.\par
\egroup

\noindent
Such a specification can be identified with the (monotonic) predicate transformer
$\psi\mapsto\hbox{``biggest such $\phi$''}$.  This point of view is that of
the $\mathbf{wp}$ calculus, introduced by Dijkstra (``$\mathbf{wp}$'' stands
for ``weakest precondition'').  Note that the specification ``goes backward in
time'': it associates to a set of final states (which we want to reach) a set
of initial states (which guarantee that we will reach our goal).\footnote{In a
previous version, interfaces also had to enjoy the property
$P(\emptyset)=\emptyset$ and $P(|X|)=|X|$. This condition doesn't interact
well with second order interpretation and has thus been dropped.}

For a complete introduction to the field of predicate transformers in relation
to specifications, we refer to~\cite{RC}.

\medbreak
In the coherence semantics, a ``point'' is a complete subgraph,\footnote{The
intuition is that a set of data is coherent iff it is pairwise coherent.}
called a \emph{clique.} Since the intuitions behind our objects are quite
different, we change the terminology.
\begin{definition}
  Let $X$ be an interface, a subset $x\subseteq|X|$ is called a \emph{seed} of
  $X$ if $x \subseteq P_X(x)$. We write $\Sd(X)$ for the collection of seeds
  of $X$.
\end{definition}
More traditional names for seeds are safety properties, or $P$-invariant
properties: if some initial state is in $x$, no matter what, after each
execution of a program satisfying specification $P$, the final state will
still be in $x$. In other words, $P$ maintains an invariant, namely
``staying in $x$''. In particular, there can be no program deadlock when
starting from $x$.

\smallbreak
The collection of cliques in the (hyper)coherent semantics forms a c.p.o.: the
sup of any directed family exists. The collection of seeds in an interface
satisfies the stronger property:

\begin{lemma}
  For any interface $X$, $\big(\Sd(X),\subseteq\big)$ is a complete
  sup-lattice.
\end{lemma}
\begin{proof}
  $\emptyset$ is trivially a seed; and by monotonicity of $P$, a union of
  seeds is a seed.
\qed
\end{proof}

The fact that seeds are closed under union may seem counter-intuitive at
first; but one possible interpretation is that we allow for non-deterministic
data.  For example, all denotational models of linear logic have an object for
the booleans: its state space is $\{t,f\}$, and the cliques are always
$\emptyset$, $\{t\}$ and $\{f\}$. The union of $\{t\}$ and $\{f\}$ is usually
not itself a clique because ``one cannot get both true and false''. However,
if one interprets union as a non-deterministic sum, then $\{t,f\}$ is a
perfectly sensible set of data.

However, nothing guarantees that a seed is the unions of all its finite
subseeds; a given seed needs not even contain any finite seed!. (The
canonical example being $P_X(x)=X$, with $X$ infinite.)
%%%>>>1

%%%%%%%%%%%%%%%%%%%%%%%%%%%%%%%%%%%%%%%%%%%%%%%%%%%%%%%%%%%%%%%%%%%%%%%%%%%%%%
\section{Constructions on Interfaces} %%%<<<1

A denotational model interprets formulas as objects in a category (and proofs
as morphisms). We thus need to define all the constructions of linear logic at
the level of interfaces. The most interesting cases are the linear negation
and the tensor product (and the exponentials, but they will be treated in
section~\ref{Section:exp}).

Note that there will always be an ``ambient'' set $A$ for predicates. We write
$\compl{x}$ for the $A$-complement of $x$.

\smallbreak
\noindent
Let $X=(|X|,P_X)$ and $Y=(|Y|,P_Y)$ be two interfaces;

\begin{definition}
  The \emph{dual} of $X$ is defined as $(|X|,P_{X}^{\bot})$ where
  $P_{X}^{\bot}(x) = \compl{P_X(\compl{x})}$. We write it $X^\bot$.  An
  \emph{antiseed} of $X$ is a seed in $X^\bot$.
\end{definition}
In terms of specifications, $a\in P^\bot(x)$ means \textsl{``if the program is
started in $a$, and if execution terminates, the final state will be in
$x$''.} If $P$ is concerned with $\mathbf{wp}$ calculus, then $P^\bot$ is
concerned with $\mathbf{wlp}$ calculus. (Weakest liberal precondition, also
introduced by Dijkstra.)

This operation of ``negation'' is the reason we do not ask for any properties
on the predicate transformer. It respects neither continuity nor commutation
properties! In many respects, this operation is not very well-behaved.

\begin{definition}
  The \emph{tensor} of $X$ and $Y$ is the interface $(|X|\times |Y|,P_X\tensor
  P_Y)$ where $P_X\tensor P_Y (r)$ is the predicate transformer
  \[
    r\mapsto \bigcup_{x\times y\subseteq r} P_X(x)\times P_Y(y) \ \hbox{.}
  \]
  We write it $X\tensor Y$.
\end{definition}
$P_X\tensor P_Y$ is the most natural transformer to construct on $|X|\times
|Y|$. It was used in~\cite{productRC} to model parallel execution of
independent pieces of programs. The intuition is the following: a program
satisfies $P_X\tensor P_Y$ if, when you start it in the pair $(a_i,b_i)\in
P_X\tensor P_Y(r)$ of initial states, the two final states will be related
through $r$. In particular, this means that execution is \emph{synchronous}:
both executions need to terminate.

\begin{definition}
  The \emph{with} of $X$ and $Y$ is the interface $(|X|+|Y|,P_X\with P_Y)$
  where $P_X\with P_Y(x,y) = \big(P_X(x),P_Y(y)\big)$.\footnote{it uses
  implicitly the fact that $\Pow(|X|+|Y|)\simeq\Pow(|X|)\times \Pow(|Y|)$} We
  write it $X\with Y$.
\end{definition}
This operation is not very interesting from the specification point of view:
it is a kind of disjoint union.

\begin{definition}
The other connectives are defined as usual:
  \begin{itemize}
  \item $\bzero=(\emptyset,\Id)$; $\top = \bzero^\bot$; $\bone=(\{*\},\Id)$; $\bot=\bone^\bot$;
  \item $X\plus Y$ (\emph{plus}) is the interface $\big(X^\bot\with Y^\bot\big)^\bot$;
  \item $X\Par Y$ (\emph{par}) is the interface $\big(X^\bot\tensor Y^\bot\big)^\bot$;
  \item $X \linearArrow Y$ is the interface $X^\bot\Par Y$.
  \end{itemize}
\end{definition}

\noindent
We have:
\begin{lemma}
  $\bot=\bone$; $\top=\bzero$ and $X\plus Y = X\with Y$.
\end{lemma}
The proof is immediate. The first two equalities are satisfied in several of
the denotational models of LL; the second one is a little less common. (For
example, it is satisfied in finiteness spaces, but in no ...-coherence
spaces.)

\medbreak
\noindent
As an application of the definitions, let's massage the definition of $A
\linearArrow B$ into something readable:

\bgroup\parindent=0pt
$(a,b) \in A \linearArrow B(r)$

\step{\Leftrightarrow}{definition}

$(a,b) \in \big(A^\bot \Par B\big)(r)$

\step{\Leftrightarrow}{definition, involutivity of $\_^\bot$}

$(a,b) \in \big(A \tensor B^\bot\big)^\bot(r)$

\step{\Leftrightarrow}{definition of $\_^\bot$}

$(a,b) \notin A \tensor B^\bot (\compl{r})$

\step{\Leftrightarrow}{definition of $\tensor$}

$\lnot\big((\exists x\times y\subseteq \compl{r})\ a \in A(x) \land b\in B^\bot(y)\big)$

\step{\Leftrightarrow}{logic}

$(\forall x\times y\subseteq \compl{r})\  a\notin A(x) \lor b\notin B^\bot(y)$

\step{\Leftrightarrow}{logic}

$(\forall  x\times y\subseteq \compl{r}) \ a\in A(x) \Rightarrow b\in B(\compl{y})$

\step{\Leftrightarrow}{lemma: $x\times y\subseteq\compl{r}$ iff $\langle r\rangle x\subseteq \compl{y}$}

$(\forall  \langle r\rangle x\subseteq \compl{y}) \ a\in A(x) \Rightarrow b\in B(\compl{y})$

\step{\Leftrightarrow}{change of variable: $y\mapsto\compl{y}$}

$(\forall  \langle r\rangle x\subseteq y) \ a\in A(x) \Rightarrow b\in B(y)$.
\egroup

\smallbreak
\noindent
From this, we derive:
\begin{lemma} \label{lem:identity}
  $(a,b) \in A \linearArrow B(r)$ iff $a\in A(x) \Rightarrow b\in
  B(\langle r\rangle x)$ for all $x\subseteq|X|$.
\\
  For any interface $X$, $\Id_{|X|}\in\Sd(X \linearArrow X)$.
\end{lemma}

The shapes of images along $X\Par Y$ are usually difficult to visualize, but we
have the following on ``rectangles'':
\begin{lemma} \label{lem:tensor_par_rect}
  Let $X$ and $Y$ be interfaces; then for all $x\subseteq |X|$ and
  $y\subseteq|Y|$ we have: $P_X\tensor P_Y(x\times y)=P_X(x)\times P_Y(y)
  \subseteq P_X\Par P_Y(x\times y)$.
\end{lemma}

\begin{proof}\parindent=0pt
That $P_X\tensor P_Y(x\times y)=P_X(x)\times P_Y(y)$ is straightforward.

\smallbreak
Suppose now $a\in P_X(x)$ and $b\in P_Y(y)$, let's show that $(a,b)\in P_X\Par
P_Y(x\times y)$:

suppose $x'\times y'\subseteq\compl{x\times y}$

\step{\Rightarrow}{claim (see below)}

$x\subseteq \compl{x'} \lor y\subseteq \compl{y'}$

\step{\Rightarrow}{monotonicity}

$a\in P_X(\compl{x'}) \lor b\in P_Y(\compl{y'})$.

\smallbreak
\textit{Claim:} $x'\times y'\subseteq\compl{x\times y} \Rightarrow x\subseteq \compl{x'} \lor
y\subseteq \compl{y'}$

\textit{Proof of claim:} suppose $\lnot(x\subseteq
\compl{x'})\land\lnot(y\subseteq \compl{y'})$

\step{\Rightarrow}{}
$x\cap x'\neq\emptyset \land y\cap y'\neq\emptyset$

\step{\Rightarrow}{}
$x\times y \cap x'\times  y'\neq\emptyset$

\step{\Rightarrow}{}
$\lnot(x'\times y'\subseteq \compl{x\times y})$.
\qed
\end{proof}

\medbreak
Furthermore, seeds in $A$ and $B$ are related to seeds in $A\tensor B$ and
$A\Par B$ in the following way:
\begin{lemma}\label{lem:seed-tensor}
  Let $A$ and $B$ be interfaces. We have:
  \begin{itemize}
  \sitem{i} if $x\in\Sd(A)$ and $y\in\Sd(B)$ then $x\times y\in\Sd(A\tensor B)$;
  \sitem{ii} if $x\in\Sd(A)$ and $y\in\Sd(B)$ then $x\times y\in\Sd(A\Par B)$.
  \end{itemize}
\end{lemma}

\begin{proof} \parindent=0pt
The first point is obvious; the second point is a direct consequence of
Lemma~\ref{lem:tensor_par_rect}.
\qed
\end{proof}
%%%>>>1

%%%%%%%%%%%%%%%%%%%%%%%%%%%%%%%%%%%%%%%%%%%%%%%%%%%%%%%%%%%%%%%%%%%%%%%%%%%%%%
\section{Linear Proofs and Seeds} %%%<<<1

The previous section gave a way to interpret any linear formula $F$ by a
interface $F^*$. (When no confusion arises, $F^*$ is written $F$.) We now
interpret linear proofs of $F$ as subsets of the state space of
$F^*$.\footnote{recall that a sequent $A_1,\dots A_n$ is interpreted by
$A_1\Par\dots A_n$ and the notation $\pi\vdash\Gamma$ means ``$\pi$ is a proof
of sequent $\Gamma$''} We refer to~\cite{LL} or the abundant literature on the
subject for the motivations governing those inference rules.

\begin{itemize}
\sitem1 If $\pi$ is \infer{}{\vdash\bone}{}
then $\pi^*=\{*\}$;

\sitem2 if $\pi$ is \infer{}{\vdash\Gamma,\top}{}
then $\pi^*=\emptyset$;

\sitem3 if $\pi$ is \infer{\pi_1 \vdash\Gamma}{\vdash\Gamma,\bot}{}
then $\pi^*=\big\{(\gamma,*)\ |\ \gamma\in \pi_1^*\big\}$;

\sitem4 if $\pi$ is \infer{\pi_1 \vdash \Gamma,A,B}{\vdash\Gamma,A\Par B}{}
then $\pi^*=\big\{\big(\gamma,(a,b)\big) \ |\ (\gamma,a,b)\in \pi_1^*\big\}$;

\sitem5  if $\pi$ is \infer{\pi_1 \vdash \Gamma,A &
\pi_2\vdash\Delta,B}{\vdash\Gamma,\Delta,A\tensor B}{}\par
\leavevmode\hfill then $\pi^*=\pi_1^*\tensor\pi_2^* = \big\{ \big(\gamma,\delta,(a,b)\big)\ |\ (\gamma,a)\in\pi_1^*
\land (\delta,b)\in\pi_2^*\big\}$;

\sitem6 if $\pi$ is \infer{\pi_1\vdash\Gamma,A}{\vdash\Gamma,A\plus B}{}
then $\pi^*=\big\{\big(\gamma,(1,a)\big)\ |\ (\gamma,a)\in \pi_1^*\big\}$;

\sitem7 if $\pi$ is \infer{\pi_1\vdash\Gamma,B}{\vdash\Gamma,A\plus B}{}
then $\pi^*=\big\{\big(\gamma,(2,b)\big)\ |\ (\gamma,b)\in \pi_1^*\big\}$;

\sitem8 if $\pi$ is  \infer{\pi_1\vdash\Gamma,A &
\pi_2\vdash\Gamma,B}{\vdash\Gamma,A\with B}{}\par
\leavevmode\hfill then $\pi^*$ is
$\big\{\big(\gamma,(1,a)\big) | (\gamma,a)\in\pi_1^*\big\} \cup\big\{(\gamma,(2,b)) | \big(\gamma,b\big)\in\pi_2^*\big\}$;

\sitem9 if $\pi$ is \infer{\pi_1\vdash\Gamma,A &
\pi_2\vdash\Delta,A^\bot}{\vdash\Gamma,\Delta}{}\par
\leavevmode\hfill then $\pi^*=\big\{(\gamma,\delta)\ |\ (\exists\,a)\ (\gamma,a)\in\pi_1^* \land (\delta,a)\in\pi_2^*\big\}$.

\end{itemize}

\noindent
This interpretation is correct in the following sense:

\begin{proposition}\label{prop:sound-MALL}
  If $\pi$ a proof of $F$, then $\pi^*$ is a seed in $F^*$.
\end{proposition}

\begin{proof}
By induction on the structure of $\pi$: we will check that seeds propagate
through the above constructions. It is mostly trivial computation, except for
two interesting cases:

\smallbreak\parindent=0pt

\sit5: suppose that $\pi_1$ is a seed in $\Gamma\Par A$ and that $\pi_2$ is a
seed in $\Delta\Par B$. We need to show that $\pi_1\tensor\pi_2 =
\big\{\big(\gamma,\delta,(a,b)\big)\ |\ (\gamma,a)\in \pi_1 \land
(\delta,b)\in \pi_2\big\}$ is a seed in the sequent
$\Gamma\Par\Delta\Par(A\tensor B)$.

Let $\big(\gamma,\delta,(a,b)\big) \in \pi_1\tensor\pi_2$

\step{\Leftrightarrow}{}

$(\gamma,a)\in\pi_1$ and $(\delta,b)\in\pi_2$

\step{\Rightarrow }{$\pi_1$ and $\pi_2$ are seeds in $\Gamma,A$ and $\Delta,B$}

$(\gamma,a)\in \Gamma,A(\pi_1)$ and $(\delta,\pi_2)\in \Delta,B(\pi_2)$.

\smallbreak

By contradiction, let $\big(\gamma,\delta,(a,b)\big) \notin \Gamma,\Delta,A\tensor B (\pi_1\tensor\pi_2)$

\step{\Rightarrow}{}

$\big(\gamma,\delta,(a,b)\big) \in \Gamma^\bot\tensor\Delta^\bot\tensor(A\tensor B)^\bot
(\compl{\pi_1\tensor\pi_2})$

\step{\Rightarrow}{for some $u\times v\times r \subseteq
\compl{\pi_1\tensor\pi_2}$:}

$ \gamma\in\Gamma^\bot(u)\land \delta\in\Delta^\bot(v)\land \smash{\hbox{$\underbrace{(a,b)\in(A\tensor B)^\bot(r)}$}}$

\step{\Rightarrow}{}

$\dots \land \Big((\forall x\times y\subseteq \compl{r})\ a\in A^\bot(\compl{x}) \lor b\in
B^\bot(\compl{y})\Big)$.

\smallbreak

In particular, define $x=\langle \pi_1\rangle u$ and $y=\langle \pi_2\rangle
v$; it is easy to show that $x\times y\subseteq\compl{r}$, so that we have
$a\in A^\bot(\compl{x})$ or $b\in B^\bot(\compl{y})$.

Suppose $a\in A^\bot(\compl{x})$: we have $\gamma\in \Gamma^\bot(u)$ and
$u\times \compl{x}\subseteq \compl{\pi_1}$ (easy lemma); so by definition,
$(\gamma,a)\in \Gamma^\bot\tensor A^\bot(\compl{\pi_1})$, \ie
$(\gamma,a)\notin \Gamma,A(\pi_1)$! This is a contradiction.

Similarly, one can derive a contradiction from $b\in B^\bot(\compl{y})$.

This finishes the proof that $\pi_1\tensor\pi_2$ is a seed of
$\Gamma,\Delta,A\tensor B$.

\smallbreak
\sit9: let $\pi_1$ be a seed in $\Gamma,A = \Gamma^\bot \linearArrow A$ and
$\pi_2$ a seed in $\Delta,A^\bot$, \ie $\pi_2^\sim$ is a seed in $A
\linearArrow \Delta$. Let's show that $\pi=\big\{(\gamma,\delta)\ |\ (\exists
a)\ (\gamma,a)\in\pi_1 \land (\delta,a)\in\pi_2\big\} = \pi_2^\sim \cdot
\pi_1$ is a seed in $\Gamma,\Delta$.

Suppose $(\gamma,\delta) \in \pi_2^\sim \cdot \pi_1$, \ie that
$(\gamma,a)\in\pi_1$ and $(a,\delta)\in\pi_2^\sim$ for some $a$. We will prove
that $(\gamma,\delta)$ is in $\Gamma,\Delta(\pi) = \Gamma^\bot \linearArrow
\Delta(\pi)$.  According to Lemma~\ref{lem:identity}, we need to show
that if $\gamma\in\Gamma^\bot(u)$ then $\delta\in\Delta(\langle \pi\rangle
u)$.

Let $\gamma\in\Gamma^\bot(u)$

\step{\Rightarrow}{$(\gamma,a)\in\pi_1\subseteq \Gamma^\bot \linearArrow A(\pi_1)$}

$a\in A\big(\langle \pi_1\rangle u\big)$

\step{\Rightarrow}{$(a,\delta)\in\pi_2^\sim \subseteq
A\linearArrow\Delta(\pi_2^\sim)$}

$\delta \in \Delta\big(\langle\pi_2^\sim\rangle\langle\pi_1\rangle u\big)$

\step{\Leftrightarrow}{}

$\delta \in \Delta\big(\langle\pi\rangle u\big)$.
\qed
\end{proof}
%%%>>>1

%%%%%%%%%%%%%%%%%%%%%%%%%%%%%%%%%%%%%%%%%%%%%%%%%%%%%%%%%%%%%%%%%%%%%%%%%%%%%%
\section{Morphisms, Categorical Structure} %%%<<<1
\label{section:category}

To complete the formal definition of a category of interfaces, we need to
define morphisms between interfaces. This is done in the usual way:

\begin{definition}
  A linear arrow from $X$ to $Y$ is a seed in $X \linearArrow Y$.
\end{definition}
Here is a nicer characterization of linear arrows from $X$ to $Y$:
\begin{lemma} \label{lem:sim}
  $r\in\Sd(X \linearArrow Y)$ iff $\langle r\rangle \big(P_X(x)\big)\subseteq
  P_Y\big(\langle r\rangle (x)\big)$ for all $x\subseteq|X|$.
\end{lemma}
\begin{proof} \parindent0pt
Suppose $r$ is a seed in $X \linearArrow Y$, let $b\in\langle r\rangle P_X(x)$

\step{\Rightarrow}{}

there is some $a$ s.t. $(a,b)\in r$ and $a\in P_X(x)$

\step{\Rightarrow}{$r$ is a seed in $X\linearArrow Y$}

$(a,b)\in P_X\linearArrow P_Y(r)$

\step{\Rightarrow}{definition of $\linearArrow$}

$b\in P_Y(\langle r\rangle x)$.

\smallbreak
Conversely, suppose $\langle r\rangle P_X(x) \subseteq P_Y\langle r\rangle
(x)$; let $(a,b)\in r$, and $a\in P_X(x)$. We have $b\in\langle r\rangle
P_X(x)$, and by hypothesis, $b\in P_Y(\langle r\rangle x)$.
\qed
\end{proof}

\begin{lemma}
  If $r\in\Sd(X \linearArrow Y)$ and $r'\in\Sd(Y \linearArrow Z)$ then
  $r'\cdot r\in\Sd(X \linearArrow Z)$.
\end{lemma}

\begin{proof}
This is the essence of point~\sit9 from Proposition~\ref{prop:sound-MALL}; or a
simple corollary to Lemma~\ref{lem:sim}.
\qed
\end{proof}

\noindent
Taken together with Lemma~\ref{lem:identity}, this makes interfaces
into a category:
\begin{definition}
  We write $\Int$ for the category with interfaces as objects and
  linear arrows as morphisms.
\end{definition}

This category is an enrichment of the usual category $\Rel$. The construction
can be summarized in the following way:

\begin{lemma}
$\Int$ is obtained by lifting $\Rel$ through the following specification
structure (\cite{specstruct}):
\begin{itemize}
\item if $X$ is a set, $\Prop_X\equiv\Pow(X)\to\Pow(X)$;
\item if $r\subseteq X\times Y$, $P\in\Prop_X$ and
$Q\in\Prop_Y$, then $P\{r\}Q$ iff $\langle r\rangle\cdot P\subseteq
Q\cdot\langle r\rangle$.
\end{itemize}
\end{lemma}

\noindent
Let's now turn our attention to the structure of this category:
\begin{lemma}
  In $\Int$, $\top$ is terminal and $\with$ is the cartesian product.
\end{lemma}
\begin{proof}This is immediate. \qed
\end{proof}

\begin{lemma}
  $\_^\bot$ is an involutive contravariant functor.
\end{lemma}

\begin{proof}
Involutivity is trivial; contravariance is only slightly trickier:
\parindent=0pt

$r$ is a seed in $A \linearArrow B$

\step{\Leftrightarrow}{Lemma \ref{lem:sim}}

$\forall x\ \langle r\rangle A(x) \subseteq B\langle r\rangle x$

\step{\Leftrightarrow}{}

$\forall x\ \compl{B\langle r\rangle x} \subseteq \compl{\langle r\rangle
A(x)}$

\step{\Leftrightarrow}{lemma: $y\subseteq \compl{\langle r\rangle x}$ iff
$\langle r^\sim\rangle y \subseteq \compl{x}$}

$\forall x\ \langle r^\sim\rangle \compl{B\langle r\rangle x} \subseteq \compl{A(x)}$

\step{\Rightarrow}{in particular, for $x$ of the form $\compl{\langle r^\sim\rangle
x}$; we have $\compl{x} \subseteq \langle r\rangle\compl{\langle r^\sim\rangle
x}$
(lemma)}

$\forall x\ \langle r^\sim\rangle B^\bot (x) \subseteq A^\bot\big(\langle
r^\sim\rangle x\big)$

\ie $r^\sim$ is a seed in $B^\bot \linearArrow A^\bot$.
The action of $\_^\bot$ on morphisms is just $\_^\sim$.
\qed
\end{proof}

\begin{corollary}
  $\Int$ is autodual through $\_^\bot$; $\bzero$ is initial; and $\plus$ is
  the coproduct.
\end{corollary}
It is now easy to see that linear arrows transform seeds into seeds,
and, in the other direction, antiseeds into antiseeds:
\begin{proposition}
  Suppose $r$ is a linear arrow from $X$ to $Y$:
  \begin{itemize}
  \sitem{i} $\langle r\rangle $ is a sup-lattice morphism from $\Sd(X)$ to $\Sd(Y)$;
  \sitem{ii} $\langle r^\sim\rangle $ is a sup-lattice morphism from $\Sd(Y^\bot)$ to
    $\Sd(X^\bot)$.
    \end{itemize}
\end{proposition}

\begin{proof}
Let $r\in\Sd(X \linearArrow Y)$ and $x\subseteq X(x)$; we
want to show that $\langle r\rangle x \subseteq Y(\langle r\rangle x)$.
\parindent=0pt

Let $b\in\langle r\rangle x$

\step{\Leftrightarrow}{}

$(\exists a)\ (a,b)\in r \land a\in x$

\step{\Rightarrow}{$r$ is a seed in $X \linearArrow Y$}

$(\exists a)\ (a,b)\in X \linearArrow Y(r) \land a\in x$

\step{\Rightarrow}{definition of $X \linearArrow Y$ with the fact that $\langle r\rangle x\subseteq \langle r\rangle x$}

$b\in Y(\langle r\rangle x)$.

Showing that $\langle r\rangle $ commutes with sups is immediate: it commutes
with arbitrary unions, even when the argument is not a seed.

The second point follows because $r^\sim\in\Sd(Y^\bot \linearArrow X^\bot)$.
\qed
\end{proof}

\begin{lemma}
  $\tensor$ [\/$\Par$] is a  categorical tensor product with neutral
  element $\bone$ [$\bot$].
\end{lemma}

\begin{proof}
We need to show the bifunctoriality of $\tensor$. This was actually proved in
the previous section (Proposition~\ref{prop:sound-MALL}, point~\sit5). The
bifunctoriality of $\Par$ follows by duality; and the rest is immediate.
\qed
\end{proof}
As a summary of this whole section, we have:
\begin{proposition}
$\Int$ is a $*$-autonomous category.
(In particular, $\Int$ is symmetric monoidal closed.)
\end{proposition}

\begin{proof} This amounts to checking trivial equalities, in particular, that
the following diagram commutes: \hskip.5cm{\small(where $d$ is the natural
isomorphism $X\simeq X^{\bot\bot}$)}

\centerline{%
\xymatrix{X \linearArrow Y \ar[rr]^{\_^\bot} \ar[drr]_{d^{-1}_X\cdot\_\cdot
d^{\phantom{-1}}_Y} & & Y^\bot \linearArrow X^\bot \ar[d]^{\_^\bot} \\ & & X^{\bot\bot}
\linearArrow Y^{\bot\bot}}}

\noindent
It is immediate because $d=\Id$ and $\_^\bot=\_^\sim$.
\qed
\end{proof}
%%%>>>1

%%%%%%%%%%%%%%%%%%%%%%%%%%%%%%%%%%%%%%%%%%%%%%%%%%%%%%%%%%%%%%%%%%%%%%%%%%%%%%
\section{Exponentials} %%%<<<1
\label{Section:exp}

The category $\Int$ is thus a denotational model for multiplicative additive
linear logic. Let's now add the exponentials $!X$ and $?X$.

Unsurprisingly, we will use finite multisets; here are the necessary
definitions and notations:
\begin{definition}
Let $S$ be a set;
\begin{itemize}
\item if $(s_i)_{i\in I}$ and $(t_j)_{j\in J}$ are finite families on $S$, say
  $(s_i)\simeq(t_j)$ iff there is a bijection $\sigma$ from $I$ to $J$ such that
  $s_{i} = t_{\sigma(i)}$ for all $i$ in $I$.
\item A {finite multiset} over $S$ is an equivalence class of $\simeq$.
  We write $[s_i]$ for the equivalence class containing $(s_i)$.
\item $\Mulf(S)$ is the collection of finite multisets over $S$.
\item Concatenation of finite families\footnote{defined on the disjoint sum
  of the different index sets} can be lifted to multisets; it is
  written $+$.
\item If $x$ and $y$ are two subsets of $S$, we write $x\pr y$ for the set $\{
  [a,b] \ |\ a\in x\land b\in y\}$. Its indexed version is written $\Pr_{i\in
  I} x_i$; it is a kind of commutative product.
\item If $U$ and $V$ are two subsets of $\Mulf(A)$, the set $\{u+v\ |\ u\in U
  \land v\in V\}$ is written $U\pr V$ (same symbol, but no confusion
  arises).
\end{itemize}
\end{definition}

\begin{definition}
For $X=(|X|,P)$, define
$!X=(\Mulf(|X|),!P)$ where
\[[a_1,\dots a_n] \in !P(U) \quad\Leftrightarrow\quad
\big(\exists (x_i)_{1\leq i\leq n}\big)\  \Pr_{i} x_i\subseteq U
\land (\forall i=1,\dots n)\,a_i\in P(x_i)\]
Let $?X = \big(!(X^\bot)\big)^\bot$.
\end{definition}
Recall that a multiset $[a_i]$ is in $\Pr x_i$ iff there is a bijection
$\sigma$ s.t. $\forall i, a_i\in x_{\sigma(i)}$.

A useful intuition is that $[a_1,\dots]\in!P(U)$ iff $[a_1,\dots]$ is in a
``weak infinite tensor'' $\bigoplus_n X^{\tensor n}(U)$. In terms of
specifications and programs, it suggests multithreading: for an initial state
$[a_1,\dots a_n]$, start $n$ occurrences of the program in the states
$a_1$,\dots $a_n$; the final state is nothing but the multiset of all the $n$
final states.\footnote{The interpretation of $!$, like that of $\tensor$ is a
synchronous operation.} The ``weak'' part means that we forget the link
between a particular final state and a particular initial state.

\smallbreak
Note that this is a ``non-uniform'' model in the sense that the web of $!X$
contains all finite multisets, not just those whose underlying set is a seed.
It is thus closer to non-uniform (hyper)coherence semantics
(see~\cite{nonUniformThomas} or~\cite{nonUniformPierre}) than to the
traditional (hyper)coherence semantics.

\medbreak\goodbreak
\noindent
Let's prove a simple lemma about the exponentials:
\begin{lemma}\label{lem:exp}
  Suppose $U\subseteq \Mulf(|A|)$:
  \begin{itemize}
  \sitem{i} $[a] \in !A(U)$ iff there is some $x$ ``included'' in $U$ (\ie
  $\forall a\in x\ [a]\in U$) s.t. $a\in A(x)$;
  \sitem{ii} $l+l' \in !A(U)$ iff there are $V\pr V'\subseteq U$
  s.t. $l\in !A(V)$ and
  $l'\in!A(V')$;
  \sitem{iii} $[a] \in ?A(U)$ iff for all $\compl{x}$ ``included'' in $\compl{U}$, $a\in A(x)$;
  \sitem{iv} $l+l' \in ?A(U)$ iff for all $\compl{V}\pr \compl{V'}\subseteq
  \compl{U}$, $l\in ?A(V)$ or $l'\in?A(V')$.
  \end{itemize}
\end{lemma}
\begin{proof}
The first point is immediate and the second is left as an exercise. The third
and last point are consequences of the definition of $?$ in terms of
$!$.
\qed
\end{proof}

\bigbreak
Define now the interpretation of proofs with exponentials:

\begin{itemize}
\sitem{10} if $\pi$ is \infer{\pi_1\vdash\Gamma,A}{\vdash\Gamma,?A}{}
then $\pi^*=\big\{\big(\gamma,[a]\big)\ |\ (\gamma,a)\in\pi_1^*\big\}$;

\sitem{11} if $\pi$ is \infer{\pi_1\vdash\Gamma}{\vdash\Gamma,?A}{}
then $\pi^*=\big\{\big(\gamma,[\,]\big)\ | \ \gamma\in\pi_1^*\big\}$;

\sitem{12} if $\pi$ is \infer{\pi_1\vdash\Gamma,?A,?A}{\vdash\Gamma,?A}{}
then $\pi^*=\big\{(\gamma,l+l')\ |\ (\gamma,l,l')\in\pi_1^*\big\}$;

\sitem{13} if $\pi$ is \infer{\pi_1\vdash?\Gamma,A}{\vdash?\Gamma,!A}{}\\
then we define $\big(\gamma_1,\ldots \gamma_l,[a_1\ldots a_n]\big) \in \pi^*$
if for each $j=1,\ldots l$, there is a partition $\gamma_j = \sum_{1\leq i\leq
n} \gamma_j^i$ and the following holds: for each $i=1,\ldots n$, $(\gamma_1^i,\ldots
\gamma_l^i,a_i)\in \pi_1^*$.
\end{itemize}

\begin{proposition} \label{prop:sound-exp}
  If $\pi$ a proof of\/ $\vdash\Gamma$, then $\pi^*$ is a seed of $\Gamma$.
\end{proposition}

\begin{proof} Points~\sit{10} and \sit{11} are immediate.
\parindent=0pt

\sit{12}: suppose $\pi_1$ is a seed $\Gamma,?A,?A$ and let
$(\gamma,l+l')$ be an element of $\pi$.

By contradiction, suppose that $(\gamma,l+l')\notin \Gamma,?A(\pi)$

\step{\Leftrightarrow}{}

$(\gamma,l+l') \in \Gamma^\bot\tensor!A^\bot(\compl{\pi})$

\step{\Leftrightarrow}{for some $u\times U\subseteq \compl{\pi}$}

$\gamma\in\Gamma^\bot(u) \land l+l'\in!A^\bot(U)$

\step{\Leftrightarrow}{Lemma~\ref{lem:exp}}

$\gamma\in\Gamma^\bot(u) \land (\exists
V\pr V'\subseteq U)\ l\in !A^\bot(V) \land l'\in!A^\bot(V')$

\step{\Rightarrow}{lemma: $u\times V\times V'\subseteq\compl{\pi_1}$}

$%(\exists u\times V\times V'\subseteq\compl{\pi_1})\ 
\gamma\in\Gamma^\bot(u) \land
l\in!A^\bot(V) \land l'\in!A^\bot(V')$

\step{\Rightarrow}{}

$(\gamma,l,l')\in\Gamma^\bot\tensor!A^\bot\tensor!A^\bot(\compl{\pi_1})$

\step{\Leftrightarrow}{}

$(\gamma,l,l') \notin \Gamma,?A,?A(\pi_1)$, which contradicts the fact that
$\pi_1$ is a seed in $\Gamma,?A,?A$.

\medbreak
\sit{13}: suppose that $\Gamma$ contains only one formula $B$.
The general case will follow from a lemma proved below
(Lemma~\ref{lem:linearIsomorphism}).
Suppose that $\pi_1$ is a seed in $?B,A$; let
$(l,[a_1,\dots a_n])$ be in $\pi$, \ie $(l_i,a_i)\in\pi_1$ for
$i=1,\dots n$, for some partition $(l_1,\dots l_n)$ of $l$.

Suppose by contradiction that $(l,[a_1\dots a_n])\notin ?B,!A(\pi)$

\step{\Leftrightarrow}{}

$(l,[a_1,\dots a_n]) \in !B^\bot\tensor?A^\bot(\compl{\pi})$

\step{\Leftrightarrow}{for some $U\times V\subseteq\compl{\pi}$}

$l\in!B^\bot(U) \land [a_1,\dots a_n]\in ?A^\bot(V)$

\step{\Rightarrow}{definition of $?A$}

$l\in!B^\bot(U) \land
\Big(\big(\forall
(x_i)\big)\ \Pr x_i\subseteq \compl{V} \Rightarrow (\exists i)\
a_i\in A(\compl{x_i})\Big)$

\step{\Leftrightarrow}{Lemma \ref{lem:exp} for $l$: for some
$(U_i)$ s.t. $\Pr_i
U_i\subseteq U$}

$(\forall i)\ l_i\in!B^\bot(U_i)\land \Big(\big(\forall
(x_i)_{i}\big)\ \Pr_{i} x_i\subseteq \compl{V} \Rightarrow (\exists i)\ \dots
%a_i\in A(\compl{x_i})\Big)
$

\step{\Rightarrow}{define $x_i = \langle \pi_1\rangle U_i$; lemma: $\Pr_{i}x_i \subseteq \compl{V}$}

$\Big((\forall i)\ l_i\in!B^\bot(U_i)\Big) \land \Big((\exists i)\ a_i\in
A(\compl{x_i})\Big)$

\step{\Rightarrow}{lemma: $U_i\times x_i\subseteq\compl{\pi_1}$}

$(\exists i)\ (\exists U_i\times x_i\subseteq\compl{\pi_1})\ l_i\in!B^\bot(U_i) \land
a_i\in A(\compl{x_i})$

\step{\Leftrightarrow}{}

$(l_i,a_i)\in !B^\bot\tensor A^\bot(\compl{\pi_1})$

\step{\Leftrightarrow}{}

$(l_i,a_i) \notin ?B,A(\pi_1)$, which contradicts the fact that $\pi_1$ is a
seed in $?B,A$.
\qed
\end{proof}

\medbreak
\begin{lemma}  \label{lem:linearIsomorphism}
  For all interfaces $X$ and $Y$, we have $!(X\with Y) = !X \tensor
  !Y$.
\end{lemma}

\begin{proof} The state spaces are isomorphic via
$\Mulf(|X|+|Y|)\simeq\Mulf(|X|)\times \Mulf(|Y|)$. We will use this
transparently, for example $l_X+l_Y\in R$ iff $(l_X,l_Y)\in R$. This is
possible because the sets are disjoint: we can always split a multiset in $x\pr
y$ into two multisets in $x$ and $y$.
(In other words: if $x\cap y=\emptyset$ then $x\pr y \simeq x\times y$.)

Notice also that $(1,a)\in X\with Y(x,y) \Leftrightarrow a\in X(x)$ so that
when considering a particular element of $X+Y(x,y)$, only one part of the
argument $(x,y)$ is really important; the other can be dropped (or replaced
with $\emptyset$).
\parindent=0pt

\smallbreak

$\subseteq$: suppose $[a_1,\dots a_n]+[b_1,\dots b_m] \in !(X\with
Y)(R)$

\step{\Leftrightarrow}{for some $(x_i)_{i=1\dots n}$ and $(y_j)_{j=1\dots m}$}

$\Pr_{i} x_i
\pr \Pr_j y_j \subseteq R \land (\forall i)\ a_i\in X(x_i) \land
(\forall j) b_j \in Y(y_j)$

\step{\Rightarrow}{define $U'=\Pr_i x_i$ and $V'=\Pr_j y_j$}

$(\exists U'\times V'\subseteq R)\ [a_i]\in !X(U') \land [b_j]\in!Y(V')$

\step{\Leftrightarrow}{}

$([a_1,\dots a_n],[b_1,\dots b_m])\in !X \tensor !Y (R)$.

\smallbreak
$\supseteq$: suppose
$([a_1,\dots a_n],[b_1,\dots b_m])\in !X \tensor !Y (R)$

\step{\Leftrightarrow}{for some $U'\times V'\subseteq R$}

$[a_1,\dots a_n]\in!X(U') \land [b_1,\dots b_m]\in!Y(V')$

\step{\Leftrightarrow}{for some $(x_i)$ \st $\Pr x_i\subseteq U'$ and $(y_j)$ \st $\Pr y_j\subseteq V'$}

$(\forall i)\ a_i\in X(x_i) \land (\forall j)\ b_j\in Y(y_j)$

\step{\Rightarrow}{$\Pr_i x_i \times 
\Pr_j y_j \subseteq U\times V$ and thus $\Pr_i x_i \pr 
\Pr_j y_j \subseteq R$}

$[a_1,\dots a_n]+[b_1,\dots b_m] \in !(X\with Y) (R)$.
\qed
\end{proof}

\noindent
This allows us to transform any sequent $?\Gamma=?B_1\Par\dots?B_n$ into
$?(B_1\plus\dots B_n)$, and thus, formally ends the proof of
Proposition~\ref{prop:sound-exp} point \sit{13}.
%%%>>>1

%%%%%%%%%%%%%%%%%%%%%%%%%%%%%%%%%%%%%%%%%%%%%%%%%%%%%%%%%%%%%%%%%%%%%%%%%%%%%%
\section{Linear Interfaces and Linear Seeds} %%%<<<1
\label{Section:linear-seeds}

What is the structure of those interfaces that come from a linear formula? The
answer is unfortunately trivial:

\begin{proposition}
  If $F$ is a linear formula, then $P_F = \Id_{\Pow|F|}$.
\end{proposition}

\begin{proof} Immediate induction. Let's treat the case of
the exponentials: \parindent=0pt

suppose $F(x)=x$; suppose moreover that $[a_1,\dots a_n]\in U$

\step{\Rightarrow}{}

$a_i\in F(\{a_i\})$ for all $i$ and $\Pr \{a_i\} = \{[a_1,\dots
a_n]\}\subseteq U$

\step{\Rightarrow}{}

$[a_1,\dots a_n]\in !F(U)$

\smallbreak
Similarly, suppose $[a_1,\dots a_n]\in !F(U)$

\step{\Rightarrow}{}

each $a_i\in F(x_i)=x_i$ for some $(x_i)$ s.t. $\Pr x_i\subseteq U$

\step{\Rightarrow}{$[a_1,\dots a_n]\in\Pr x_i$}

$[a_1,\dots a_n]\in U$.
\qed
\end{proof}

In particular, every subset of $|F|$ is a clique and an anticlique: the
situation is thus quite similar to the purely relational model. In the
presence of atoms however, interfaces become much more interesting.

Adding atoms is sound because the proof of Proposition~\ref{prop:sound-MALL}
doesn't rely on the particular properties of interfaces. Note that
we need to introduce a general axiom rule and its interpretation:
\begin{itemize}
\sitem{14} if $\pi$ is \infer{}{\vdash X,X^\bot}{}
then $\pi^*=\Id_{|X|}=\{(a,a)\ |\ a\in|X|\}$.
\end{itemize}
\noindent
This is correct in the sense that $\pi^*$ is always a clique in $X\Par X^\bot$.

With such atoms, the structure of linear interfaces gets non
trivial.\footnote{We can extend this to a model for $\PI^1$ logic, and even to
full second order, see~\cite{PTSecondOrder}.} For example, let's consider the
following atom $X=\big(\{\m,\p\},P\big)$ defined by:
\begin{itemize}
\item $P(\emptyset)=\emptyset$ and $P(|X|)=|X|$;
\item $P(\{\p\})=\{\m\}$ and $P(\{\m\})=\{\p\}$.
\end{itemize}
This is the simplest example of an interesting interface, and
corresponds to a ``switch'' specification. (Interpret $\m$ as ``off'' and $\p$
as ``on''.)
\begin{lemma}
if $P$ is the above specification:
\begin{itemize}
\sitem{i} $P^\bot = P$;
\sitem{ii} $P\cdot P=\Id$;
\sitem{iii} $\Sd(X) = \big\{\emptyset,\{\p,\m\}\big\}$;
\sitem{iv} $\big\{(\p,\m),(\m,\p)\big\} \in \Sd(X\tensor X)$.
\end{itemize}
\end{lemma}
\begin{proof}
This is just trivial computation...
\qed
\end{proof}

Point \sit{iv} shows in particular that a seed in $X\tensor Y$ needs not contain a
product of seeds in $X$ and $Y$. (Compare with
Lemma~\ref{lem:seed-tensor}.)

\smallbreak
The hierarchy generated from this single interface is however still relatively
simple: call a specification \emph{deterministic} if it commutes with
non-empty unions and intersections.
\begin{lemma}
  Let $F$ be any specification constructed from the above $P$ and the linear
  connectives. Then $F$ is deterministic. Moreover, $F$ is of the form
  $\langle f\rangle$ where $f$ is an obvious bijection on the state space of
  $F$.\footnote{where $\langle f\rangle(x) = \{f(a)\ |\ a\in x\}$}
\end{lemma}

\medbreak
A less trivial (in the sense that it is not deterministic) specification is
the following: if $X$ is a set, $\magic_X(x)=X$. In terms of programming, the
use of the $\magic$ command allows to reach any predicate, even the empty one!

\begin{lemma}
  $\Id_{|X|} \varsubsetneq \magic_X \linearArrow \magic_X(\Id_{|X|})$ if
  $X\neq\emptyset$.
\end{lemma}
Thus we cannot strengthen the definition of seeds to read ``$x=P(x)$'' without
imposing further constraints on our specifications. It is still an open
question to find a nice class of predicate transformers for which it would be
possible. (However, considerations about second order seem to indicate that
strengthening the definition of seeds in such a way is not a good idea.)

\medbreak
In the case with atoms, because the structure of seeds (sup-lattice) is quite
different from the structure of cliques in the ...-coherent model (domain), it
is difficult to relate seeds and cliques. In particular, a seed needs not be a
clique (since the union of arbitrary cliques is not necessarily a clique); and
a clique needs not be a seed (since a subset of a seed is not necessarily a
seed).
%%%>>>1

%%%%%%%%%%%%%%%%%%%%%%%%%%%%%%%%%%%%%%%%%%%%%%%%%%%%%%%%%%%%%%%%%%%%%%%%%%%%%%
\section*{Conclusion} %%%<<<1

%As stated in section \ref{section:category}, interfaces are a
%particular example of specification structure over $\Rel$.\footnote{We do
%believe however that this specific example is particularly interesting...} The
%methodology from~\cite{specstruct} can be applied to obtain the
%category~$\Int$: in turn
%\begin{itemize}
%  \item put $\Prop_X \equiv \Pow(X)\to\Pow(X)$; \kern.5cm{\small(monotonic
%    predicate transformers)}
%  \item put $x \vDash P$ iff $x\subseteq P(x)$; \kern1.015cm{\small($x$ is a
%    safety property of $P$)}
%  \item define $P^\bot$ and $P\tensor Q$.
%\end{itemize}
%Considering the definition of the linear negation, it is ``surprising'' that the
%(degenerated) linear structure of $\Rel$ lifts to $\Int$.

One aspect which was not really mentioned here is the fact that linear arrows
from $A$ to $B$ are equivalent to the notion of \emph{forward data refinement}
(Lemma~\ref{lem:sim}) from the refinement calculus. In particular, a linear
proof of $A \linearArrow B$ is a proof that specification $B$
\emph{implements} specification $A$. It would interesting to see if any
application to the refinement calculus could be derived from this work.  In
the same direction, trying to make sense of the notions of \emph{backward data
refinement}, or of \emph{general data refinement} in terms of linear logic
could prove interesting.\footnote{A data refinement from specification $F$ to
specification $G$ is a predicate transformer $P$ s.t. $P\cdot F \subseteq
G\cdot P$; a forward [resp. backward] data refinement is a data refinement
which commutes with arbitrary unions [resp. arbitrary intersections].}

\medbreak
The fact that this model is degenerate in the propositional case is
disappointing, but degeneracy disappear when we consider $\PI^1$ logic, and
\textit{a fortiori} when we consider full second-order
(see~\cite{PTSecondOrder}). The point of extending this propositional model to
$\PI^1$ is to remove the dependency on specific valuations for the atoms
present in a formula.

\medbreak
One the interesting consequences of this work is that a a proof of a
formula~$F$ gives a guarantee that the system specified by the formula~$F$ can
avoid deadlocks seems to point toward other fields like process calculi and
similar models for ``real'' computations.  This direction is currently being
pursued together with the following link with the differential lambda-calculus
(\cite{difflambda}): one property of this model which doesn't reflect any
logical property is the following; we have a natural transformation
$A\linearArrow !A$ called \emph{co-dereliction}, which has a natural
interpretation in terms of differential operators on formulas
(see~\cite{diffnet}).  Note that such a natural transformation forbids any
kind of completeness theorem, at least as far as ``pure'' linear logic is
concerned.

\smallbreak
%%%>>>1

\bigbreak
\bibliographystyle{splncs}
%\bibliography{denot_PT} %%%<<<1

%%%>>>1
\end{document}